\documentclass[11pt]{article}

\usepackage{amsmath,amssymb,amsbsy,amsfonts,amsthm,latexsym,
            amsopn,amstext,amsxtra,euscript,amscd,color}
            
\PassOptionsToPackage{hyphens}{url}\usepackage{hyperref}
    
\usepackage[capbesideposition=outside,capbesidesep=quad]{floatrow}

\newtheorem{theorem}{Theorem}[section]

\newtheorem{thm}{Theorem}[section]

\newtheorem{cor}[thm]{Corollary}

\newtheorem{prop}[thm]{Proposition}
\newtheorem{rem}[thm]{Remark}
\newtheorem{defn}[thm]{Definition}

\def\+{\oplus}

\def\F{{\mathbb F}}
\def\Z{{\mathbb Z}}

\def\F{{\mathbb F}}

\def\Z{{\mathbb Z}}

\def\00{{\bf 0}}
\def\11{{\bf 1}}
\def\+{\oplus}

\def\\{\cr}
\def\({\left(}
\def\){\right)}

\newcommand{\cardinality}[1]{\# #1}

\usepackage[colorinlistoftodos,prependcaption,textsize=tiny]{todonotes}

\providecommand{\newoperator}[3]{%
  \newcommand*{#1}{\mathop{#2}#3}}

\newoperator{\FD}{\mathrm{FD}}{\nolimits}

\begin{document}
\title{\bf Higher Order $c$-Differentials}
\author{Aaron Geary$^1$, Marco Calderini$^2$, Constanza Riera$^3$, \and Pantelimon~St\u anic\u a$^1$ 
 \vspace{0.15cm} \\
 \small  $^1$ Applied Mathematics Department, 
 Naval Postgraduate School, \\
 \small Monterey, USA;  {\tt \{aaron.geary, pstanica\}@nps.edu}\\
\small $^2$ Department of Informatics, University of Bergen\\
\small Postboks 7803, N-5020, Bergen, Norway;
\small {\tt Marco.Calderini@uib.no}\\
\small$^3$Department of Computer Science,\\
\small Electrical Engineering and Mathematical Sciences,\\
\small   Western Norway University of Applied Sciences,\\
\small  5020 Bergen, Norway; {\tt csr@hvl.no}
  }

\date{September 6, 2021}
\maketitle

\begin{abstract}
In~\cite{EFRST20}, the notion of $c$-differentials was introduced as a potential expansion of differential cryptanalysis against block ciphers utilizing substitution boxes.     
Drawing inspiration from the technique of higher order differential cryptanalysis, in this paper we propose the notion of higher order $c$-derivatives and differentials and investigate their properties.   Additionally, we consider how several classes of functions, namely the multiplicative  inverse function and the Gold function, perform under higher order $c$-differential uniformity.
\end{abstract}
\newline


\section{Introduction and background}
\label{sec1}

The newly proposed $c$-differentials \cite{EFRST20} modify the traditional differential cryptanalysis technique by applying a multiple ``$c$" to one of the outputs of an S-box primitive $F$.  If an input pair $(x,x+a)$ with difference ``$a$" results in an output pair $(F(x),F(x+a))$ with difference $b=F(x+a)-F(x)$, then the couple $(a,b)$ is the traditional {\em differential} traced throughout a cipher.  A differential that appears with a high probability is used as the basis of a classical differential attack~\cite{BS91}.   The new $c$-differential uses a modified output pair of $(cF(x),F(x+a))$, and the new output difference is then $b=F(x+a)-cF(x)$.  Similar to other extensions and modifications of differential cryptanalysis, $c$-differentials have been shown to result in higher probabilities than traditional differentials for some functions \cite{EFRST20}, \cite{MRSZY20}, thus potentially resulting in attacks against ciphers that are resistant against other forms of differential cryptanalysis.

The introduction of $c$-differentials and the corresponding $c$-differential uniformity ($c$DU) has been met with substantial interest.  Researchers have since submitted multiple papers (see~\cite{BC20,MRSZY20,SGGRT20,WLZ20,ZH20}, just to cite only a few of these works) further exploring the topic.  These include investigations of the $c$DU of various classes of functions, finding functions with low $c$DU, construction and existence results on the so-called perfect $c$-nonlinear and almost perfect $c$-nonlinear functions, and generalizations of cryptographic properties to include the new $c$-differential.  

In this paper, we continue this investigation by considering the extension of $c$-differentials into higher order.  This is motivated by the extension of the original differential cryptanalysis technique into higher order differential cryptanalysis (\cite{La94}, \cite{Kn94}).  In contrast with the traditional higher order derivatives of Boolean or $p$-ary functions, the $c$-derivative and higher order $c$-derivative do not always reduce the degree of a function.  However, in the same spirit as traditional higher order differentials, higher order $c$-differentials have the potential to allow for a better trace of multiple differences through an encryption scheme, and any resistance against such higher order differentials with large probabilities furthers the case of a cipher's security.  

The rest of the paper is organized as follows. In Section \ref{secprelim} we provide the necessary notation and definitions to introduce the higher order $c$-derivative and investigate its properties in Section \ref{secHO}.  In Sections \ref{secinv} and \ref{gold} we consider specific higher order $c$-differential cases of the inverse function and Gold function over finite fields.   Section \ref{secsum} summarizes our findings.

\section{Preliminaries}
\label{secprelim}
We introduce here some basic notations and definitions on Boolean and $p$-ary functions (where $p$ is an odd prime); the reader can consult~\cite{Bud14,CH1,CS17,MesnagerBook,Tok15} for more on these objects. For a positive integer $n$ and $p$ a prime number, we denote by $\F_p^n$ the $n$-dimensional vector space over $\F_p$, and by $\F_{p^n}$ the  finite field with $p^n$ elements, while $\F_{p^n}^*=\F_{p^n}\setminus\{0\}$ will denote the multiplicative group.  We call a function from $\F_{p^n}$ (or $\F_p^n$) to $\F_p$  a {\em $p$-ary  function} on $n$ variables. For positive integers $n$ and $m$, any map $F:\F_{p^n}\to\F_{p^m}$ (or, $\F_p^n\to\F_p^m$)  is called a {\em vectorial $p$-ary  function}, or {\em $(n,m,p)$-function}. If $p=2$ the function is called a vectorial Boolean function.  In any characteristic, when $m=n$ the function $F$ can be uniquely represented as a univariate polynomial over $\F_{p^n}$ (using some identification, via a basis, of the finite field with the vector space) of the form
$
F(x)=\sum_{i=0}^{p^n-1} a_i x^i,\ a_i\in\F_{p^n},
$
whose {\em algebraic degree}, denoted by $\deg(F)$, is then the largest weight  in the $p$-ary expansion of $i$ (that is, the sum of the digits of the exponents $i$ with $a_i\neq 0$). To (somewhat) distinguish between the vectorial and single-component output, we shall use upper/lower case to denote the functions.
 
Given a $(n,m,p)$-function $F$, the derivative of $F$ with respect to~$a \in \F_{p^n}$ is the $(n,m,p)$-function
\[ D_{a}F(x) =  F(x + a)- F(x), \mbox{ for  all }  x \in \F_{p^n}.\]
The distribution of the derivatives of an $(n,m,p)$-function used in an S-box is important.  If we let  $\Delta_F(a,b)=\cardinality{\{x\in\F_{p^n} : F(x+a)-F(x)=b\}}$, then we call the quantity
$\delta_F=\max\{\Delta_F(a,b)\,:\, a,b\in \F_{p^n}, a\neq 0 \}$ the {\em differential uniformity} of $F$.

The $i$-th derivative of $F$ at $(a_1,a_2,\ldots,a_i)$ is defined recursively as  
\[ D_{a_1,\ldots,a_i}^{(i)}F(x) = D_{a_i}(D_{a_1,\ldots,a_{i-1}}^{(i-1)}F(x)).
\]

The new $c$-differential, which applies a multiplier to one of the outputs, immediately leads to a modified derivative.  For an $(n,m,p)$-function $F$, and $a\in\F_{p^n},b\in\F_{p^m}$, and $c\in\F_{p^m}$, the ({\em multiplicative}) {\em $c$-derivative} of $F$ with respect to~$a \in \F_{p^n}$ is the  function
\[
 {_c}D_{a}F(x) =  F(x + a)- cF(x), \mbox{ for  all }  x \in \F_{p^n}.
\]

Equipped with this new $c$-derivative, a new $c$-autocorrelation function was defined in \cite{SGGRT20}, and several cryptographic properties of $(n,m,p)$-functions were generalized.  That work continues in this paper as we extend the $c$-derivative into higher order, investigate its properties, and then analyze the higher order $c$-differential uniformity of several functions.

\section{Higher order $c$-differentials}
\label{secHO}

Inspired by the concept of higher order derivatives of functions between Abelian groups and their applications to cryptography in \cite{La94}, we propose the following definition.

\begin{defn}  Let $F:\F_{p^n}\rightarrow\F_{p^m}$ be an $(n,m,p)$-function.  The $i$-th $c$-derivative of $F$ at $(a_1,a_2,\ldots,a_i)$ is 
\[
 {_c}D_{a_1,\ldots,a_i}^{(i)}F(x) =  \;{_c}D_{a_i}({_c}D_{a_1,\ldots,a_{i-1}}^{(i-1)}F(x)),
\]
where ${_c}D_{a_1,\ldots,a_{i-1}}^{(i-1)}F(x)$ is the $(i-1)$-th derivative of $F$ at $(a_1,a_2,\ldots,a_{i-1})$.
\end{defn}
This implies the $0$-th $c$-derivative is the function $F$ itself and the 1st $c$-derivative is the $c$-derivative defined in Section \ref{secprelim}.  Notice that, when $c=1$, we recover the traditional $(n,m,p)$-function higher order derivative.  

Before we explore these new higher order derivatives we need to ensure several basic properties carry over from the traditional (i.e. $c=1$) case.  First, we see that the sum rule holds.  That is, that the $c$-derivative of a sum is a sum of the $c$-derivatives. 
\begin{equation} \label{eq2}
\begin{split}
{_c}D_a(F+G)(x) & =F(x+a)+G(x+a)-c(F(x)+G(x))\nonumber\\
& =F(x+a)-cF(x) + G(x+a)-cG(x)\nonumber\\
& =\;{_c}D_aF(x)+{_c}D_aG(x).
\end{split}
\end{equation}

A product rule exists for the traditional derivative, $D_a(FG)(x)=F(x+a)D_aG(x)+D_aF(x)G(x)$. We find something similar with the $c$-derivative,
\begin{equation} \label{eq3}
\begin{split}
{_c}D_a(FG)(x) & =F(x+a)G(x+a)-cF(x)G(x)\nonumber\\
& =F(x+a)\left(G(x+a)-cG(x)\right)+\left((F(x+a)-F(x)\right)cG(x)\nonumber\\
& =F(x+a)\;{_c}D_aG(x)+\;{_c}D_aF(x)\;c\,G(x).
\end{split}
\end{equation}

Now we consider the higher order $c$-derivatives. When $i=2$ we have 
\allowdisplaybreaks
\begin{align*}
 {_c}D_{a_1,a_2}^{(2)}F(x) =&  {_c}D_{a_2}({_c}D_{a_1}F(x))\\
 =& {_c}D_{a_2}(F(x + a_1)- cF(x))\\
 =& F(x+a_1+a_2)-cF(x+a_2) - c(F(x+a_1)-cF(x))\\
 =& F(x+a_1+a_2)-cF(x+a_2)- cF(x+a_1)+c^2F(x).
 \end{align*}

Taking another iteration,  we have 
\allowdisplaybreaks
\begin{align*}
 {_c}D_{a_1,a_2,a_3}^{(3)}F(x) & = F(x+a_1+a_2+a_3)\\
& -c\left[F(x+a_1+a_2)+F(x+a_1+a_3)+F(x+a_2+a_3)\right]\\
& + c^2\left[F(x+a_1)+ F(x+a_2)+F(x+a_3)\right] - c^3F(x).
\end{align*}

We see a similar pattern to Proposition 1 in \cite{La94}, albeit with the additional complication of powers of $c$, and we find the following identity:
\allowdisplaybreaks
\begin{align*}
F(x+a_1+a_2+a_3) & = \;{_c}D_{a_1,a_2,a_3}^{(3)}F(x)\\
& + c\left[{_c}D_{a_1,a_2}^{(2)}F(x)+{_c}D_{a_1,a_3}^{(2)}F(x)+{_c}D_{a_2,a_3}^{(2)}F(x)\right]\\
& + c^2\left[{_c}D_{a_1}(F(x))+{_c}D_{a_2}(F(x))+{_c}D_{a_3}(F(x))\right]\\
& + c^3F(x).
\end{align*}

The pattern holds in general, as we now show.

\begin{thm}\label{higher}
Let $F$ be an $(n,m,p)$-function with ${_c}D_{a_1,\ldots,a_i}^{(i)}F(x)$ the $i$-th $c$-derivative of $F$ at $(a_1,a_2,\ldots,a_i)$.  Then 
%
\begin{equation}\label{eq1}
F\left(x+\sum_{i=1}^n a_i\right) = \sum_{i=0}^{n} \sum_{1\leq j_1 < \ldots < j_i \leq n } c^{n-i} \;{_c}D_{a_{j_1},\ldots,a_{j_i}}^{(i)}F(x) .
\end{equation}
\end{thm}

\begin{proof}
Equation \eqref{eq1} can also be written as 
\[
F\left(x+\sum_{i=1}^n a_i\right) = \;{_c}D_{a_1,\ldots,a_n}^{(n)}F(x) + \sum_{i=0}^{n-1} \sum_{1\leq j_1 < \ldots < j_i \leq n-1 } c^{n-1-i} \;{_c}D_{a_{j_1},\ldots,a_{j_i}}^{(i)}F(x),
\]
which implies
\[
{_c}D_{a_1,\ldots,a_n}^{(n)}F(x) = F\left(x+\sum_{i=1}^n a_i\right) - \sum_{i=0}^{n-1} \sum_{1\leq j_1 < \ldots < j_i \leq n-1 } c^{n-1-i} \;{_c}D_{a_{j_1},\ldots,a_{j_i}}^{(i)}F(x).
\]
We proceed by induction.  For $n=1$ we see \eqref{eq1} follows directly from the definition and $n=2,3$ can be seen in the discussion before the theorem.  Assuming Equation \eqref{eq1} holds for $n-1$, we have
\begin{align*}
&{_c}D_{a_1,\ldots,a_n}^{(n)}F(x) = \;{_c}D_{a_n}\left(\;{_c}D_{a_1,\ldots,a_{n-1}}^{(n-1)}F(x)\right)  \\
& =  \;{_c}D_{a_n}\left(F\left(x+\sum_{i=1}^{n-1} a_i\right)-\sum_{i=0}^{n-2} \sum_{1\leq j_1 < \ldots < j_i \leq n-2 } c^{n-2-i} \;{_c}D_{a_{j_1},\ldots,a_{j_i}}^{(i)}F(x)\right)\\
& = F\left(x+\sum_{i=1}^n a_i\right) - c F\left(x+\sum_{i=1}^{n-1} a_i\right) \\
& - \;{_c}D_{a_n}\left(\sum_{i=0}^{n-2} \sum_{1\leq j_1 < \ldots < j_i \leq n-2 } c^{n-2-i} \;{_c}D_{a_{j_1},\ldots,a_{j_i}}^{(i)}F(x)\right).
\end{align*}
We apply the induction hypothesis to $cF(x+a_1+\cdots+a_{n-1})$, and noticing the last double sum is composed of all the $c$-derivatives that include $a_n$, we have 
\begin{align*}
& F\left(x+\sum_{i=1}^n a_i\right) - c F\left(x+\sum_{i=1}^{n-1} a_i\right) \\
& - \;{_c}D_{a_n}\left(\sum_{i=0}^{n-2} \sum_{1\leq j_1 < \ldots < j_i \leq n-2 } c^{n-2-i} \;{_c}D_{a_{j_1},\ldots,a_{j_i}}^{(i)}F(x)\right)\\
& = F\left(x+\sum_{i=1}^n a_i\right) -c\left(\sum_{i=0}^{n-2} \sum_{1\leq j_1 < \ldots < j_i \leq n-2 } c^{n-2-i} \;{_c}D_{a_{j_1},\ldots,a_{j_i}}^{(i)}F(x)\right)\\
& - \left(\sum_{i=0}^{n-1} \sum_{1\leq j_1 < \ldots < j_i \leq n-1 } c^{n-1-i} \;{_c}D_{a_{j_1},\ldots,a_{j_i},a_{n}}^{(i)}F(x)\right) \\
& = F\left(x+\sum_{i=1}^n a_i\right) - \left(\sum_{i=0}^{n-2} \sum_{1\leq j_1 < \ldots < j_i \leq n-2 } c^{n-1-i} \;{_c}D_{a_{j_1},\ldots,a_{j_i}}^{(i)}F(x)\right)\\
& + \left(\sum_{i=0}^{n-1} \sum_{1\leq j_1 < \ldots < j_i \leq n-1 } c^{n-1-i} \;{_c}D_{a_{j_1},\ldots,a_{j_i},a_{n}}^{(i)}F(x)\right)\\
& =  F\left(x+\sum_{i=1}^n a_i\right) - \left(\sum_{i=0}^{n-1} \sum_{1\leq j_1 < \ldots < j_i \leq n-1 } c^{n-1-i} \;{_c}D_{a_{j_1},\ldots,a_{j_i}}^{(i)}F(x)\right).
\end{align*}
The claim is shown.
\end{proof}

While we have shown several properties of the $c$-derivative closely align with the traditional derivative, one key property does not follow.  A fundamental property of traditional derivatives is that the degree of a polynomial function is reduced by at least one for every derivative taken.  That is, $\deg(D_aF) \leq \deg(F) - 1$.  This is not always true in the case of $c$-derivatives when $c \neq 1$.  For example, consider the linearized monomial  $F(x)=x^{p^k}$ over $\F_{p^n}$ with $k$ an integer between 0 and $n$.  This function has degree 1 (recall that the $p$-ary weight of $p^k$ is 1) and the $c$-derivative of $F$ at $a$ is $(x+a)^{p^k}-cx^{p^k}=(1-c)x^{p^k}+a^{p^k}$, which is also of degree 1 for all $c \neq 1$. Thus, the reduction of degree is not a general property of the $c$-derivative. 

We will now show that higher order $c$-derivatives are invariant under permutation of the $a_i$'s.

\begin{prop}
Let $F:\mathbb{F}_{p^n}\to\mathbb{F}_{p^n}$, denote $[t]:=\{1,\ldots,t\}$, and let $|I|$ be the cardinality of the subsets ${I\subseteq [t]}$. Then, $$
{_c}D_{a_1,\ldots,a_t}^{(t)}F(x)=\sum_{I\subseteq [t]}(-c)^{t-|I|}F\left(x+\sum_{i\in I} a_i\right).
$$ In particular, for any permutation $\pi$ of $\{1,\ldots,t\}$ we have ${_c}D_{a_1,\ldots,a_t}^{(t)}F(x)= {{_c}D^{(t)}_{a_{\pi(1)},\ldots,a_{\pi(t)}}}F(x)$.
\end{prop}
\begin{proof}
It is easy to see that  ${_c}D_{a_1,a_2}^{(2)}F(x)=\sum_{I\subseteq [2]}(-c)^{2-|I|}F(x+\sum_{i\in I} a_i)
$,
and by induction we get
$$
\begin{aligned}
&{_c}D_{a_1,\ldots,a_t}^{(t)}F(x)={_c}D_{a_1,\ldots,a_{t-1}}^{(t-1)}F(x+a_t)-c {_c}D_{a_1,\ldots,a_{t-1}}^{(t-1)}F(x)\\
&=\sum_{I\subseteq [t-1]}(-c)^{(t-1)-|I|}F\left(x+a_t+\sum_{i\in I} a_i\right)-c\sum_{I\subseteq [t-1]}(-c)^{(t-1)-|I|}F\left(x+\sum_{i\in I} a_i\right)\\
&=\sum_{\substack{I'\subseteq [t]\\a_t\in I'}}(-c)^{(t-1)-(|I'|-1)}F\left(x+\sum_{i\in I'} a_i\right)+\sum_{\substack{I'\subseteq [t]\\a_n\notin I'}}(-c)^{t-|I'|}F\left(x+\sum_{i\in I'} a_i\right)\\
&=\sum_{I\subseteq [t]}(-c)^{t-|I|}F\left(x+\sum_{i\in I} a_i\right).
\end{aligned}
$$
From this we can see that permuting the elements $a_i$ does not change the value of the higher order $c$-derivative.
\end{proof}
Another easy fact to check is that for power functions the $t$-order c-differential uniformity can be computed by considering $a_1=1$. 
\begin{prop}
Let $F(x)=x^d$ on $\F_{p^n}$ and denote by $_c\Delta(a_1,\ldots,a_t;b)=\#\{x\,:\,{_c}D_{a_1,\ldots,a_t}^{(t)}F(x)=b\}$. Then, assuming that not all $a_i$'s are zero \textup{(}and without loss of generality we can assume that $a_1\ne 0$\textup{)}, ${_c}\Delta(a_1,\ldots,a_t;b)={_c}\Delta\left(1,a_2/a_1\ldots,a_t/a_1;b/(a_1)^d\right)$.
\end{prop}

One of the key findings of higher order derivatives of binary functions is that if the $i$ inputs are not linearly independent, then the $i$th derivative is exactly 0.  That is, if $a_1, a_2, \ldots , a_i$ are linearly dependent, then $D_{a_1,\ldots,a_i}^{(i)} F(x)= 0$.  This limits the number of pairs that can be attempted in a higher order differential attack to the dimension of the vector space and reduces the combinations of differences that can be traced simultaneously.  However, this property, and therefore the limits, do not apply for higher order $c$-derivatives when $c \neq 1$, which  can be seen by considering the definition of the $c$-derivative.  If we let $a=0$, then 
\[
 {_c}D_{0}F(x) =  F(x + 0) - cF(x) = (1-c)F(x).
\]

Thus, even in the extreme case of zero difference between the input pairs, the $c$-derivative results in a nonzero function.  In higher order $c$-derivatives this property remains true.  For example the 2nd $c$-derivative has the form 
\[
{_c}D_{a_1,a_2}^{(2)}F(x) = F(x+a_1+a_2)-cF(x+a_2)- cF(x+a_1)+c^2F(x).
\]

Even if $a_1=a_2$ (and thus linearly dependent), the 2nd $c$-derivative is not identically zero due to the introduction of the $c$ multiplier.  This fact increases the input (or output) differences that can be traced through an encryption scheme and potentially increases the vulnerability of a cipher if a $c$-differential attack is realized in the future.

As for the 1st order derivative, for an $(n,m,p)$-function we can introduce the {\em $t$-order $c$-differential uniformity} of $F$ at $c$ to be  
\[
\displaystyle
_{c}\delta^{(t)}_{F}=\max_{a_1,\ldots,a_t\in\mathbb{F}_{p^n}}  {_c}D_{a_1,\ldots,a_t}^{(t)}F(x)=b
\]
(if $c=1$, not all $a_i$'s are allowed to be zero).  If $t=1$, we recover the $c$-differential uniformity ${_c}\delta^{(1)}_{F}={_c}\delta_{F}$, as defined in~\cite{EFRST20}.

\begin{prop} 
\label{lowerbound} 
For any $(n,m)$-function $F$, any order $t\in\Z_+$ (positive integers), and any $c\neq1$, the $t$-order $c$-differential uniformity of $F$ is greater than or equal to its $(t-1)$-order $c$-differential uniformity,
$\delta^{(t)}_{F,c}\geq \delta^{(t-1)}_{F,c}$.
\end{prop}

\begin{proof} For any $(n,m)$-function $F$ and any $c\neq1$, and taking $a_t=0$, we obtain that ${_c}D_{a_1,a_2,\ldots,a_{t-1},0}^{(t)}F(x) =(1-c)\, {_c}D_{a_1,a_2,\ldots,a_{t-1}}^{(t-1)}F(x)$, which implies our claim.
\end{proof}

\section{The inverse function}
\label{secinv}

Next we consider an example of a higher order $c$-derivative and compare it to the traditional higher order derivative (i.e. when $c=1$).   The function we investigate is the multiplicative inverse function over finite fields of characteristic 2, a popular function used in S-boxes that can be represented by a monomial $F:\F_{2^n}\rightarrow \F_{2^n}$, $F(x)=x^{2^n-2}$.  In \cite{EFRST20} the authors investigate the first $c$-derivative of this function, focusing on the newly defined $c$-differential uniformity property.  Specifically, they count the maximum number of solutions of ${_c}D_{a}F(x) = b$ for some $a,b,c \in \F_{2^n}$.  In this section we count solutions to the second order $c$-differential equation  ${_c}D_{a_1,a_2}^{(2)}F(x) = b$ with $c, a_1,a_2,b \in \F_{2^n}$. 

The traditional ($c=1$) second order differential spectrum of the inverse function over $\F_{2^n}$ was recently investigated in \cite{TMM21}.  It was shown that for $n \geq 3$ the number of solutions to $D_{a_1,a_2}x^{2^n-2} = b$ is in the set $\{0,4,8\}$ and that there are multiple $a_1, a_2, b$ that provide 8 solutions for $n \geq 6$.  

With this understanding of the behavior of the second traditional derivative of the inverse function, we now consider the second $c$-derivative of the inverse function and compare the two. Starting with ${_c}D_{a_1,a_2}^{(2)}F(x) = b$, we have 
\begin{equation} \label{eq4}
(x+a_1+a_2)^{2^n-2}+c(x+a_2)^{2^n-2}+c(x+a_1)^{2^n-2}+c^2x^{2^n-2}=b.
\end{equation}

It was shown in \cite{EFRST20} that the inverse function has a bijective first $c$-derivative when $c=0$.  Later in \cite{SGGRT20} we showed that any permutation will have a bijective (i.e. balanced) first $c$-derivative when $c=0$.  For the second $c$-derivative of the inverse function, when $c=0$, we get $(x+a_1+a_2)^{2^n-2}=b$.  If $b=0$, $x=a_1+a_2$ is the only solution.  If $b \neq 0$, then $x \neq a_1+a_2$ and $\frac{1}{x+a_1+a_2}=b$.  Thus, $x=\frac{1}{b}+a_1+a_2$ is the only solution and we see when $c=0$ the second $c$-derivative of the inverse function is a bijection, as is in the case of the first $c$-derivative when $c=0$.  

In fact, from Theorem \ref{higher}, we see that the $n$th $0$-derivative of a function $F$ is $F(x+a_1+a_2+\cdots+a_n)$, which is bijective if and only if $F$ is bijective.  Thus permutations have bijective $n$-th $c$-derivatives for all $n$ when $c=0$.  

\vspace{3 mm}
For $c \neq 0$, we consider multiple cases.  

\noindent
 {\em Case $(i)$.} Let $a_1=a_2$. Recall this leads to a trivial result in traditional derivatives.  Equation~\eqref{eq4} becomes 
  \allowdisplaybreaks
$
 x^{2^n-2}+c^2x^{2^n-2}=b \text{, that is, }
 (1+c^2)x^{2^n-2}=b.
 $
When $b=0$, $x=0$ is the only solution. If $b \neq 0$, then $x \neq 0$ and $\frac{1+c^2}{x}=b$ gives us one solution $x=\frac{1+c^2}{b}$.

\noindent
{\em Case $(ii)$.} $a_1 \neq a_2$, and $x = a_1, a_2, a_1+a_2$, or 0.  Equation~\eqref{eq4} becomes, respectively, 
\begin{align*}
& a_2^{2^n-2}+c(a_1+a_2)^{2^n-2}+c^2a_1^{2^n-2}=b \text{, or,}\\
& a_1^{2^n-2}+c(a_1+a_2)^{2^n-2}+c^2a_2^{2^n-2}=b \text{, or,}\\
& ca_2^{2^n-2}+ca_1^{2^n-2}+c^2(a_1+a_2)^{2^n-2}=b \text{, or,}\\
& (a_1+a_2)^{2^n-2}+ca_2^{2^n-2}+ca_1^{2^n-2}=b.
\end{align*}
When $c=1$ all four of these solutions are the same and can be true simultaneously.  However, when $c \neq 1$ we cannot combine all four of these solutions.  In fact, the most that can be combined are two.  Consider the solutions for $a_1$ and $a_2$ (the first two above).  If we could combine these, then we would have,
\[
a_2^{2^n-2}+c(a_1+a_2)^{2^n-2}+c^2a_1^{2^n-2} = a_1^{2^n-2}+c(a_1+a_2)^{2^n-2}+c^2a_2^{2^n-2},
\]
which simplifies to
\begin{align*}
 a_2^{2^n-2}+c^2a_1^{2^n-2}= a_1^{2^n-2}+c^2a_2^{2^n-2}\text{, or,}
 (1+c^2)a_1^{2^n-2}=(1+c^2)a_2^{2^n-2}.
\end{align*}
For $c \neq 1$ (which is an assumption throughout), $a_1$ must equal $a_2$ which is not true in this case.  Therefore, the solutions cannot be combined and we have that no more than three solutions can be true simultaneously.

Now we consider the possibility of combining $x=0$ with $x=a_1+a_2$.  This gives us
\[
ca_2^{2^n-2}+ca_1^{2^n-2}+c^2(a_1+a_2)^{2^n-2}=(a_1+a_2)^{2^n-2}+ca_2^{2^n-2}+ca_1^{2^n-2},
\]
which simplifies to
$c^2(a_1+a_2)^{2^n-2}=(a_1+a_2)^{2^n-2}.$
This is only true when $c=1$ or $a_1=a_2$, neither of which are allowed in this case. From this we immediately see that we can combine at most two of the solutions in Case $(ii)$.
There are only at most four values of $c$ which allow the combination of two of these solutions. As we have seen, there are only four possible combinations (which in some cases might be equal): $x=0$ and $x=a_1$, $x=0$ and $x=a_2$, $x=a_1+a_2$ and $x=a_1$, and $x=a_1+a_2$ and $x=a_2$.

Let $x=0$ and $x=a_1$ be both solutions of the equation above. Then, 
\[ (a_1+a_2)^{2^n-2}+ca_2^{2^n-2}+ca_1^{2^n-2}=a_2^{2^n-2}+c(a_1+a_2)^{2^n-2}+c^2a_1^{2^n-2}.
\]
Rearranging terms, we arrive at 
\[
(1+c)(a_1+a_2)^{2^n-2}+(1+c)a_2^{2^n-2}+c(1+c)a_1^{2^n-2}=0,
\]
which, since $c\neq1$, simplifies to
$ (a_1+a_2)^{2^n-2}+a_2^{2^n-2}+ca_1^{2^n-2}=0.$

If $a_1=0$, then $x=0$ and $x=a_1$ are the same solution, so we can assume that $a_1\neq0$. If $a_2=0$, we arrive at the equation $(1+c)a_1^{2^n-2}=0$, which only has the forbidden solutions $c=1$ or $a_1=0$. We can then assume that $a_1a_2\neq0$. The equation becomes then
$c(a_1+a_2)a_2+a_1^2=0$,
which has a single solution $c_0=\frac{a_1^2}{(a_1+a_2)a_2}$. It is easy to see that $c_0=0$ if and only if $a_1=0$. However, it is possible to obtain that $c_0=1$ if $a_1^2+a_2^2+a_1a_2=0$, which is achievable only if $n$ is even and $a_1=a_2\omega$ or $a_1=a_2\omega^2$, where $\F_4=\{0,1,\omega,\omega^2\}$. As long as $n\geq5$, we can always chose valid $a_1,a_2$ to ensure that $c_0\neq1$. 
By symmetry, $x=0$ and $x=a_2$ give $c_1=\frac{a_2^2}{(a_1+a_2)a_1}$, with the same conditions as $x=0$ and $x=a_1$.

Now, if $x=a_1+a_2$ and $x=a_1$, then 
\[ca_2^{2^n-2}+ca_1^{2^n-2}+c^2(a_1+a_2)^{2^n-2}=a_2^{2^n-2}+c(a_1+a_2)^{2^n-2}+c^2a_1^{2^n-2},
\]
which, by rearranging, becomes
$ (1+c)a_2^{2^n-2}+c(1+c)a_1^{2^n-2}+c(1+c)(a_1+a_2)^{2^n-2}=0,$
and, since $c\neq1$, this can be simplified to
$a_2^{2^n-2}+ca_1^{2^n-2}+c(a_1+a_2)^{2^n-2}=0.
$
If $a_1=0$, then we have the case $x=0, x=a_2$. If $a_2=0$, we do not have two different solutions. We can then assume $a_1a_2\neq0$. Then, the equation is equivalent to
$(a_1+a_2)a_1+ca_2^2=0$,
which has the solution $c_2=\frac{a_1(a_1+a_2)}{a_2^2}$. It is easy to see that $c_2\neq0$, and that $c\neq1$ under the same conditions as for $x=0$ and $x=a_1$. By symmetry, $x=a_1+a_2$ and $x=a_1$ gives $c_3=\frac{a_2(a_1+a_2)}{a_1^2}$.

\noindent
{\em Case $(iii)$.} $a_1 \neq a_2$, $x \neq a_1, a_2, a_1+a_2$, or 0.  Equation~\eqref{eq4} becomes
\[
\frac{1}{x+a_1+a_2} + \frac{c}{x+a_1} + \frac{c}{x+a_2} + \frac{c^2}{x} = b.
\]
Multiplying through by $(x+a_1+a_2)(x+a_1)(x+a_2)x$, collecting and rearranging terms, we arrive at
\begin{align}
bx^4 + & \left(1+c^2\right)x^3 + \left(a_2 + ca_2 + ba_2^2 + a_1 + ca_1 + ba_1^2 + ba_1a_2\right)x^2\nonumber\\
 + & \left(a_1a_2+ca_2^2 + c^2 a_2^2 + ca_1^2 + c^2 a_1^2 + c^2 a_1a_2 + ba_1a_2^2 + ba_1^2a_2\right)x\label{eq5}\\
&\qquad\qquad\qquad \qquad\qquad\qquad \qquad\qquad\quad +c^2a_1a_2 (a_1+a_2) = 0.\nonumber
\end{align} 

This quartic polynomial has at most four solutions when $b \neq 0$ and at most three when $b=0$. 
Without the four guaranteed solutions from Case $(ii)$, we cannot reach the 8 solutions possible when $c=1$.  This means that, as in the case of the first $c$-derivative of the inverse function, when $c \neq 1$ the differential counts \emph{decreases} from the traditional case. In fact, when combining Cases $(ii)$ and $(iii)$, a maximum of 6 solutions is possible, if $c\neq1$.  

\begin{theorem}
Let $n\ge 4$, and $F(x)=x^{2^n-2}$ over $\mathbb{F}_{2^n}$. Then, for any $c\in\mathbb{F}_{2^n}\setminus\{1\}$, ${_c}\delta_F^{(2)}\le 6$.
\end{theorem}

Some computations to demonstrate our findings are captured in Table~\ref{table1}. From here, we see that the maximum is attainable for $n=8,9$. We conjecture that it is attainable for all $n\geq8$.

  \begin{table}[ht]
            \begin{center}
               \begin{tabular}{|c||c|c|c|}
               \hline
               $n$ & $c=1$ & $c \neq 1$ & $c=0$ \\
               
               \hline
               4 & 4 & 5 & 1\\
               5 & 4 & 4 & 1\\
               6 & 8 & 5 & 1\\
               7 & 8 & 5 & 1\\
               8 & 8 & 6 & 1\\
               9 & 8 & 6 & 1\\
            
               \hline
               \end{tabular}
            \end{center}
            \caption{Maximum number of solutions to ${_c}D_{a_1,a_2}^{(2)}x^{2^n-2} = b$}
            \label{table1}
      \end{table}
      
      \section{The Gold function}
      \label{gold}

The $c$-differential uniformity for quadratic functions was characterized in~\cite{BC20}, where the authors focus on the PcN and APcN case though their proof applies for the general case. In particular from Theorem 3.1 in \cite{BC20} we can obtain the following result.

\begin{thm}
Let $q=p^h$, $F:\mathbb{F}_{p^n}\to\mathbb{F}_{p^n}$ be a quadratic function given by 
$$
\sum_{i,j}c_{i,j}x^{q^i+q^j}+\sum_{l}c_lx^{p^l}.
$$
Let $\delta=\max_{b\in\mathbb{F}_{p^n}}|F^{-1}(b)|$. Then, for $c\in\mathbb{F}_{p^{\gcd(n,h)}}\setminus\{1\}$, ${_c}\delta_F^{(t)}=\delta$.
\end{thm}

\begin{proof}
From the proof of Theorem 3.1 in \cite{BC20}, we have that for $c\in\mathbb{F}_{p^{\gcd(n,h)}}\setminus\{1\}$ the $c$-derivative ${_c}D_aF(x)$ equals
$$
(1-c)F\left(x+\frac{a}{1-c}\right)+F(a)-(1-c)F\left(\frac{a}{1-c}\right)=(1-c)F\left(x+\frac{a}{1-c}\right)+\beta,
$$
with $\beta=F(a)-(1-c)F\left(\frac{a}{1-c}\right)$.
From this we can easily see that for any $a_1,\ldots,a_t$ the higher-order $c$-derivative of $F$ is 
$$
{_c}D_{a_1,\ldots,a_t}F(x)=(1-c)^tF\left(x+\frac{a_1+\cdots+a_t}{1-c}\right)+\beta',
$$
for some constant $\beta'$ depending on the $a_i$'s,
and the claim follows.
\end{proof}
 
%
%

\begin{thm}
\label{Gold2nd}
  Let $F(x)=x^{p^k+1}$ on $\mathbb{F}_{p^n}$, and $1\neq c\in \mathbb{F}_{p^n}$. Then, the maximum number of solutions of ${_c}D_{a_1,a_2}^{(2)}F(x) = b$ is $p^{\gcd(k,n)}+1$, and there exist some $b,a_1,a_2$ such that this bound is obtained.
  \end{thm}
     \begin{proof} The equation ${_c}D_{a_1,a_2}^{(2)}F(x) = b$ is
 \[
 (x+a_1+a_2)^{p^k+1}-c(x+a_1)^{p^k+1}-c(x+a_2)^{p^k+1}+c^2x^{p^k+1}=b, 
 \]
 which renders 
  $(1-2c+c^2)x^{p^k+1}+(a_1+a_2)(1-c)x^{p^k}+(a_1+a_2)^{p^k}(1-c)x+(a_1+a_2)^{p^k+1}-c(a_1^{p^k+1}+a_2^{p^k+1})=b.$
 Since $c\neq1$, this yields
  {\small
  $$x^{p^k+1}+\frac{a_1+a_2}{1-c}x^{p^k}+\frac{(a_1+a_2)^{p^k}}{1-c}x+\frac{(a_1+a_2)^{p^k+1}-c(a_1^{p^k+1}+a_2^{p^k+1})}{(1-c)^2}=\frac{b}{(1-c)^2}.$$  
  }
 Taking $a_2=-a_1$, this equation becomes $x^{p^k+1}=\frac{b}{(1-c)^2},$ which has at most $\gcd(p^k+1,p^n-1)$ solutions, and exactly $\gcd(p^k+1,p^n-1)$ solutions for some $b$. This implies that the maximum number of solutions to ${_c}D_{a_1,a_2}^{(2)}F(x) = b$ is lower bounded by $\gcd(p^k+1,p^n-1)$ (also derived from Proposition~\ref{lowerbound}).
 
 Taking $a_1\neq a_2$, and $x=y-\frac{a_1+a_2}{1-c}$, we can write this equation as 
 $$
 y^{p^k+1}+(a_1+a_2)^{p^k}\frac{(1-c)^{p^k-1}-1}{(1-c)^{p^k}}y+\frac{a_1^{p^k}a_2+a_1a_2^{p^k}-b}{(1-c)^2}=0.
 $$
 If $(1-c)^{p^k-1}=1$, then we obtain the equation
 $y^{p^k+1}+\frac{a_1^{p^k}a_2+a_1a_2^{p^k}-b}{(1-c)^2}=0$, which has at most $\gcd(p^k+1,p^n-1)$ solutions, and exactly $\gcd(p^k+1,p^n-1)$ solutions for some $b$.  
 
 If $(1-c)^{p^k-1}\neq1$, and $b=a_1^{p^k}a_2+a_1a_2^{p^k}$, this equation can be written as
 $y\left(y^{p^k}+(a_1+a_2)^{p^k}\frac{(1-c)^{p^k-1}-1}{(1-c)^{p^k}}\right)=0,$
 which has at most $\gcd(p^k,p^n-1)+1=2$ solutions, and exactly  2 solutions for some $a_1,a_2$.
 
 Otherwise, taking $z=\alpha y$, where $\alpha^{p^k}=(a_1+a_2)^{p^k}\frac{(1-c)^{p^k-1}-1}{(1-c)^{p^k}}$ (note that the $p^k$-root exists since, for any $p$, $\gcd(p^k,p^n-1)=1$ and $x^{p^k}$ is therefore a permutation on $\F_{p^n}$), we get the equation
$
  z^{p^k+1}+z+\beta=0,
 $
 where $\beta=\alpha^{-(p^k+1)}\ \frac{a_1^{p^k}a_2+a_1a_2^{p^k}-b}{(1-c)^2}$.

It is easy to see that the equation fulfills the conditions imposed in~\cite{Bluher04}, where it is shown that there are either 0,1,2 or $p^d+1$ solutions to this equation, where $d=\gcd(k,n)$. Taking $p>2$, and if $m=\frac{n}{d}$ is even, we have (using \cite{EFRST20}) that $\gcd(p^k+1,p^n-1)=p^d+1$. From the results above, this bound is obtained, and we have exactly $p^d+1$ solutions for the general equation for some $a_1, a_2$. 

If $m$ is odd, for $p\geq 2$, then $m>2$ since otherwise $n=k$. Then, by \cite{Bluher04}, the amount of values of  $\beta$ such that there are $p^d+1$ solutions to the equation is nonzero. Since $\beta$ is linear on $b$, the maximum number of solutions taken over all $b$ is $p^d+1$.

One still has to study the case where $p=2$ and $m=2$. In this case, by \cite{EFRST20}, we have that $\gcd(2^k+1,2^n-1)=\frac{2^{\gcd(2k,n)}-1}{2^{\gcd(k,n)}-1}=\frac{2^{2d}-1}{2^{d}-1}=2^d+1$. So, in that case as well, by the previous results we obtain the bound $p^d+1$.
 \end{proof}

     The following corollary implies that, for any odd characteristic, we can always obtain functions whose higher differential uniformity is 2, regardless of the order of differentiation.
      
\begin{cor}
\label{Goldspecial}
 Let $F(x)=x^{p^k+1}$, and let $1\neq c\in\F_{p^{\gcd(k,n)}}$. Then, the maximum number of solutions to ${_c}D_{a_1,a_2,\ldots,a_t}^{(t)}F(x) = b$ is $\gcd(p^k+1,p^n-1)$.
 \end{cor}
\begin{rem}
 It is not difficult to modify the argument to obtain the same outcome as above for the $t$-order derivative taken with respect to different $c$'s.
\end{rem}

\section{Summary and further comments}
\label{secsum}

In this paper we investigate  higher order $c$-differentials, noting that traditional derivatives are a special case of our extension (i.e. when $c=1$).  
We also look at the specific case of the inverse function over fields of even characteristic  and the Gold function over any characteristic. While many properties of higher order $c$-differentials are preserved  from the traditional higher order derivative,  a key difference arises in that the higher order $c$-derivatives do not require linearly independent input differences. Thus, the higher order $c$-derivatives we have introduced could potentially allow the use of more input pairs (for encryption or decryption), which in turn could lead to differentials with higher probabilities than traditional higher order differential attacks or the new $c$-differential attack using one derivative.

 \end{document}